\documentclass[12pt,centertags,reqno]{amsart}

\usepackage{latexsym}
\usepackage[english]{babel}
\usepackage[T1]{fontenc}
\usepackage[numbers]{natbib}
\usepackage{amssymb}
\usepackage{fancyhdr}
\usepackage{url}
\usepackage{hyperref}
\usepackage{verbatim}
\usepackage{leftidx}
\usepackage{epsfig}
 \usepackage{graphicx}
 \usepackage{epstopdf}

\usepackage{mathrsfs, mathtools}
\usepackage{stmaryrd}

\usepackage{marvosym}
\mathtoolsset{showonlyrefs}

\newtheorem{dfn}{Definition}[section]

\newtheorem{theorem}{Theorem}[section]

\newtheorem{remark}{Remark}[section]

\usepackage{lineno}
\def \p {\mathbf{p}}
\def \bmu {\boldsymbol{\mu}}
\def \btheta {\boldsymbol{\theta}}

\def \F {\mathcal F}
\def \bF {\mathbb F}
\newcommand{\ud}{\mathrm d}
\newcommand{\ds}{\displaystyle}
\newcommand{\esp}[2][\mathbb E] {#1\left[#2\right]}
\def \I {{\mathbf 1}}
\def \P {\mathbf P}

\def \R {\mathbb R}

\begin{document}

%\markboth{S.~Altay, K.~Colaneri, Z.~Eksi}
%{}

%\catchline{}{}{}{}{}
%%%%%%%%%%%%%%%%%%%%%%%%%%%%%%%%%%%%%%%%%%%%%%%%%%%%%%%%%%%%%%%%%%%%

\title{Pairs Trading under Drift Uncertainty and Risk Penalization}

\author[S. Altay]{S\"{u}han Altay}
\address{S\"{u}han Altay, \textrm{Department of Financial and Actuarial Mathematics, Vienna University of Technology, Wiedner Hauptstrasse 8-10, 1040 Vienna, Austria}}
\email{altay@fam.tuwien.ac.at}
\author[K.~Colaneri]{Katia Colaneri}
\address{Katia Colaneri, \textrm{School of Mathematics, University of Leeds,  LS2 9JT Leeds, UK.}}\email{k.colaneri@leeds.ac.uk}
\author[Z.~Eksi]{Zehra Eksi}
\address{Zehra Eksi, \textrm{Institute for Statistics and Mathematics, WU-University of Economics and Business, Welthandelsplatz 1, 1020, Vienna, Austria.}}\email{zehra.eksi@wu.ac.at}

\begin{abstract}
In this work, we study a dynamic portfolio optimization problem related to pairs trading, which is an investment strategy that matches a long position in one security with a short position in another security with similar characteristics. The relationship between pairs, called a spread, is modeled by a Gaussian mean-reverting process whose drift rate is modulated by an unobservable continuous-time, finite-state Markov chain. Using the classical stochastic filtering theory, we reduce this problem with partial information to the one with full information and solve it for the logarithmic utility function, where the terminal wealth is penalized by the riskiness of the portfolio according to the realized volatility of the wealth process. We characterize optimal dollar-neutral strategies as well as optimal value functions under full and partial information and show that the certainty equivalence principle holds for the optimal portfolio strategy. Finally, we provide a numerical analysis for a toy example with a two-state Markov chain.\\


\end{abstract}

%\keywords{Pairs trading; Regime-switching; Partial information.}
\maketitle

\section{Introduction}
 Pairs trading is an investment strategy that attempts to capitalize on market inefficiencies arising from imbalances between two or more stocks. This kind of strategy involves a long position and a short position in a pair of similar stocks that have moved together historically. Examples of such pairs can be given: ExxonMobil and Royal Dutch and Shell for the oil industry, or Pfizer and GlaxoSmithKline for the pharmaceutical industry. The underlying rationale of pairs trading is to buy the underperformer, and sell the overperformer, in anticipation that the security that has performed badly will make up for loss in the coming periods, perhaps even overperform the other, and vice-versa. For this reason, it is also classified as a convergence or mean-reversion strategy. The pair of stocks is selected in a way that it forms a mean-reverting portfolio referred to as the \emph{spread}. By forming an appropriate spread, pairs traders try to limit the directional risk that arises from the market's up or down movements by simultaneously going long on one stock and short in another. Since market risk is mitigated, profits depend only on the price changes between the two stocks and they can be realized through a net gain on the spread. Therefore, one can also see pairs trading in the class of \emph{market-neutral} trading strategies. To achieve market neutrality, traders can choose corresponding strategies so that the resulting portfolio has zero (CAPM) beta, hence it is beta-neutral. Alternatively, one can use a \emph{dollar-neutral} strategy, which is investing an equal dollar amount in each stock. However, we should remark that market neutrality does not imply either risk-free return or arbitrage in the classical sense. The risk inherited in pair strategies is different from the risk in investment strategies involving only a long or short position in a specific stock or market. Indeed, pairs trading is a form of statistical arbitrage, which can be defined broadly as a long-horizon trading strategy that generates riskless profits asymptotically (see \citet{hogan2004testing} for the definition of the statistical arbitrage and \citet{akyildirim2016} for the existence of statistical arbitrage for pairs trading strategies). As it is empirically documented by \citet{gatev2006pairs}, coupled with a simple pairs selection algorithm, such statistical arbitrage strategies may yield average annualized excess returns of up to 11 percent, which still remains profitable after compensated by the most conservative transaction costs.

In this work, we consider the portfolio optimization problem of a trader with a logarithmic utility from risk penalized terminal wealth investing in a pair of assets whose dynamics have a certain dependence structure in a Markov regime-switching model. More precisely, we model the spread process (log-price differential) as an Ornstein--Uhlenbeck process with a partially observable Markov modulated drift. Our motivation for modeling the drift of the spread and drifts of both assets as a function of an unobservable finite-state Markov chain has certain advantages. Firstly, drifts of financial assets are hardly constant and observable, especially if we think of the convergence-type investment strategies that are usually valid for longer periods. Secondly, although pairs are selected in such a way that they have similar characteristics, the dynamics of the spread between them might be prone to different regimes. For example, if one leg of the pair is selected to be listed in an index such as the S\&P 500  while the other is not, this might increase the demand for the one that is listed. Hence, that would eventually increase the level of the spread, at least until the one listed in the index is deleted from the index or the other leg of the pair is also added. Moreover, in reality, it is difficult to observe or characterize both microstructure (market-based) or macrostructure (economy-wide) state variables changing with respect to different regimes. That would necessitate using a partial information framework to model such state processes.

Numerous studies analyze portfolio selection problems in a full or partial information and/or Markov regime-switching framework, see, for example, \citet{zhou2003markowitz}, \citet{bauerle2004portfolio}, and \citet{sotomayor2009explicit} for the full information case with Markov regime switching or \citet{rieder2005portfolio}, \citet{frey2012portfolio}, and \citet{bjork2010optimal} for the partial information case. However, to the best of our knowledge, identification of optimal pairs trading strategies in a Markov-modulated setting under partial information is new.

Our proposed model is an extended version of the model given by \citet{mudchanatongsuk2008optimal}, who found the optimal pairs trading strategies in a dollar-neutral setting for an investor with power utility. Although investing equal dollar amount (as a proportion of wealth) in pairs seems to be restrictive, it is meaningful when CAPM betas of the selected stocks are very close to each other. Our model extends the work of \citet{mudchanatongsuk2008optimal} by allowing partially observed Markov-modulated drifts both for the price processes and the spread,  hence enabling them to change with respect to different conditions. As the second extension, to find the optimal trading strategies, we use a risk penalized terminal wealth as it is suggested in Section 2.22 of \citet{rogers2013optimal}. By penalizing the terminal wealth according to the realized volatility of the wealth process, the investor hopes to prevent the pairs trader pursuing risky strategies. Using risk penalization seems to be appropriate in pairs trading as most such strategies are executed by hedge funds and proprietary trading houses, which engage in high-risk transactions on behalf of investors. Risk penalization effectively increases the risk aversion of the trader and makes her take a less risky position. Apart from certain mathematical convenience, our choice of logarithmic utility function can be justified on several financial grounds. Firstly, although an investor can choose any utility function, representing her risk tolerance, a repetitive situation such as the one reflected in mean-reversion type trading strategies tends to force the utility function into the one that is close to a logarithmic one. For instance, in the power utility case it can be shown in a very simple example that too aggressive or too conservative choices for the risk-aversion parameter imply unrealistic preferences such as betting on strategies that have large losses with high probability and hence not suitable if the investor is focused in a long sequence of repeated trials, see e.g., Chapter 15 of  \citet{luenbergerinvestment}). This can only be alleviated when the risk-aversion parameter $\gamma$, in power utility \footnote{$U(x)=\frac{x^\gamma}{\gamma}$ for $\gamma\leq 1.$}, is close to zero, behaving more like the logarithmic utility. Therefore, we can argue that utility functions that are close to the logarithmic ones are appropriate for our setting. Secondly, by penalizing the terminal wealth with the realized volatility of the portfolio and using logarithmic utility, we can capture the intertemporal risk factor in our model more easily with just one parameter.

Although both the empirical and theoretical literature on pairs trading has been growing, published research on optimal portfolio problem is rather limited. \citet{mudchanatongsuk2008optimal} solve the stochastic control problem for pairs trading with power utility for terminal wealth. \citet{tourin2013dynamic} develop an optimal portfolio strategy to invest in two risky assets and the money market account, assuming that log-prices are co-integrated, as in the option pricing model of \citet{duan2004option}. \citet{cartea2016algorithmic} extend \citet{tourin2013dynamic} to allow the investor to trade in multiple co-integrated assets and provide an explicit closed-form solution of the dynamic trading strategy while assuming that the drift of asset returns consists of an idiosyncratic and common drift component. \citet{lee2016pairs} solve the optimal pairs trading problem within a power utility setting, where the drift uncertainty is modeled by a continuous mean-reverting process. It is also worth mentioning here the work of \citet{elliott2005pairs}, which proposes a pairs trading strategy based on stochastic filtering of a mean-reverting Gaussian Markov chain for the spread, which is observed in Gaussian noise.

Apart from identification of optimal trading strategies through utility maximization from terminal wealth, there is also recent literature on optimal liquidation and optimal (entry-exit) timing strategies related to pairs trading. For example, studies by \citet{ekstrom2011optimal}, \citet{larsson2013optimal}, and \citet{zeng2014pairs} focus on how to liquidate optimally a pairs trade by incorporating stop-loss thresholds. Moreover, \citet{leung2015optimal} study an optimal double-stopping problem to analyze the timing for starting and subsequently liquidating the position, subject to transaction costs, and \citet{lei2015costly} analyze a multiple entry-exit problem of a pair of co-integrated assets. An extensive list of references and a literature review on pairs trading and statistical arbitrage can be found in the recent survey paper by \citet{krauss2016statistical}.

To sum up, our contributions in this article can be stated as follows. First, we characterize the optimal dollar-neutral strategies both in full and partial information settings with risk-penalized terminal wealth for a log-utility trader and show that optimal strategies are dependent on both the correlation between two assets and the mean-reverting spread. The effect of risk-penalization on optimal strategies is an increase in risk-aversion uniformly in a constant proportion that is not dependent on time. Second, we characterize the optimal value function via Feynman--Kac formula. Third, using the innovations approach, we provide filtering equations that are necessary to reduce the problem with partial information to the one with full information. A nice feature of the solution in the partial information setting is that the optimal strategy is a linear function of the filtered state and hence it can be considered as a projection of the full information one on the investor's information filtration.

We also present numerical results for a toy example with a two-state Markov chain in both full and partial information settings. Our analysis shows that average data does not contain sufficient information to obtain the optimal value for the pairs trading problem for logarithmic utility preferences. This result is in contrast with the one for the classical portfolio optimization problem with Markov modulation, see  Section B in \citet*{bauerle2004portfolio}). Furthermore, our toy example suggests that there is always a gain from filtering due to the convexity arising from using filtered probabilities instead of constant ones.

The remainder of the paper is organized as follows. Section \ref{sec:setting} introduces the model. In Section \ref{sec:full_info} we analyze the portfolio optimization problem in a full information setting. In Section \ref{sec:partial_info} we solve the utility maximization problem under partial information. In Section \ref{sec:2stateMC}, we provide the numerical analysis of our toy example with a two-state Markov chain. We conclude with Section \ref{sec:discussion} and give proofs and technical results in Appendix.

\section{The model}\label{sec:setting}
\noindent We consider a finite time interval $[0,T]$ and a continuous-time finite-state Markov chain $Y$ defined on the filtered probability space $(\Omega,\mathcal{G},\mathbb{G},\mathbf{P})$, where $\mathbb{G}=(\mathcal{G}_t)_{t\geq 0}$ is the global filtration that satisfies the usual conditions; all processes we consider here are assumed to be $\mathbb{G}$-adapted. Suppose $Y$ has the state space $\mathcal{E}=\{e_1,e_2,...,e_K\}$ where, without loss of generality, we assume that $e_k$ is the basis column vector of $\mathbb{R}^K$. $Y$ has the intensity matrix $Q=(q^{ij})_{{i,j}\in\{1,\dots,K\}}$ and its initial distribution is denoted by $\Pi=(\Pi^1,\cdots,\Pi^K)$. The  semimartingale decomposition of $Y$ is given by
\begin{align}
Y_t=Y_0+\int_0^tQ^{\top}Y_s\ud s+M_t,
\end{align}
for every $t\in[0,T]$, where $M$ is an $(\mathbb{G},\bf P)$-martingale.

We consider a market with a risk-free asset and two stocks. We assume that the dynamics of the risk-free asset is given by
\begin{align}\label{eq:S0}
\ud S^{(0)}_t=rS^{(0)}_t\,\ud t, \quad S_0^{(0)}>0,
\end{align}
where $r\in\mathbb{R}$ is the risk-free interest rate. The stocks have prices $S^{(1)}$ and $S^{(2)}$, and the price process of the first stock is assumed to follow a Markov-modulated diffusion given by

\begin{align}\label{S1}
\frac{\ud S^{(1)}_t}{S^{(1)}_t}=\mu(Y_t)\,\ud t+\sigma\, \ud W^{(1)}_t, \quad S^{(1)}_0>0,\end{align}
with $\sigma>0$ and where $W^{(1)}$ is a $\mathbb{G}$-Brownian motion independent of $Y$. Since the Markov chain takes values in a finite state space we have that for every $t\in[0,T]$, $\mu(Y_t)=\bmu Y_t$ with $\bmu=(\mu_1,\dots,\mu_K)^{\top}$ and $\mu_i=\mu(e_i)\in \mathbb{R}$ for every $i\in \{1,\dots,K \}$.

It is assumed that the spread $S_t=\log{S^{(1)}_t}-\log{S^{(2)}_t}$, $t\in[0,T]$, follows a Markov-modulated Ornstein--Uhlenbeck process:
\begin{align}\label{S}
\ud S_t=\kappa(\theta(Y_t)-S_t)\,\ud t+\eta\, \ud W_t,\quad S_0\in\mathbb{R},
\end{align}
where $\kappa>0$ and $\eta>0$, $W$ is a $\mathbb{G}$-Brownian motion with $\langle W^{(1)},W\rangle_t=\rho t$, $\rho\in(-1,1)$, and $\theta(Y_t)=\btheta Y_t$, $t\in[0,T]$ with $\btheta=(\theta_1,\dots,\theta_K)^{\top}$ and $\theta_i=\theta(e_i)\in \mathbb{R}$ for every $i\in \{1,\dots,K \}$. It follows from \eqref{S1} and \eqref{S} that
\begin{align}
\frac{\ud S^{(2)}_t}{S^{(2)}_t}=\left(\mu(Y_t)-\kappa(\theta(Y_t)-S_t)+\frac{1}{2}\eta^2-\rho\sigma\eta\right)\ud t+\sigma\, \ud W^{(1)}_t-\eta \,\ud W_t, \quad S_0^{(2)}>0.
\end{align}
Let $X$ be the value of a self-financing portfolio and let $h^{(1)}$ and $h^{(2)}$ denote fractions of the wealth invested in
$S^{(1)}$ and $S^{(2)}$, respectively.
\paragraph{\bf{Admissible Investment Strategies.}} We consider \textit{dollar-neutral} pairs trading strategies. This corresponds to take $h^{(1)}$ and $h^{(2)}$ such that
\begin{equation}\label{eq:A1}
  h^{(1)}_t=-h^{(2)}_t, \quad t\in[0,T].
\end{equation}
In the sequel we are going to use the notation $h=h^{(1)}$. Note that $h_t\in\mathbb{R}$ for every $t\in[0,T]$ and the portfolio weight on the risk-free asset is always $1$. In order to ensure that the wealth process is well defined, we consider investment strategies that satisfy
\begin{equation}\label{eq:A2}
  \mathbb{E} \left[  \int_0^T h_u^2 \, \ud u \right] <\infty.
\end{equation}
\begin{dfn}
A $\mathbb{G}$-progressive self-financing investment strategy which satisfy \eqref{eq:A1} and \eqref{eq:A2} is called an \emph{admissible investment strategy}. We denote the set of admissible strategies by $\mathcal{A}$.
\end{dfn}
\noindent For every $h \in \mathcal{A}$, the dynamics of the pairs-trading portfolio is given by
\begin{align}\label{wealthdyn}
\frac{\ud X^h_t}{X^h_t}=\left(h_t\left(\kappa(\theta(Y_t)-S_t)-\frac{\eta^2}{2}+\rho\sigma\eta \right)+r  \right)\,\ud t+h_t \eta \,\ud W_t,\quad X_0^h>0.
\end{align}
Notice that for a given $h\in\mathcal{A}$, $X^h$ is a controlled process. In what follows, for the sake of notational simplicity we suppress $h$ dependency and write $X$ instead of $X^h$.
The objective of the trader is to maximize expected utility from terminal wealth. However, in the risk-penalized setting, see Section 2.22 of \citet{rogers2013optimal}, the goal is to prevent the trader from pursuing risky strategies at the expenses of the investor. The investor agrees to pay the trader at time $T$ the risk-penalized amount
\begin{align}Z_T=X_T\exp \left( -\frac{1}{2}\varepsilon \int_0^T  \eta^2h_s^2\, \ud s \right), \quad  \varepsilon \geq 0.\end{align}
Hence the terminal value of the wealth process is `discounted' by its realized volatility. It follows from It\^{o}'s formula that the dynamics of $Z$ is given by:
\begin{align}
\label{Cd}
\frac{dZ_t}{Z_t}=\left(h_t\left(\kappa(\theta(Y_t)-S_t)-\frac{\eta^2}{2}+\rho\sigma\eta \right)+r -\frac{\varepsilon\eta^2 h_t^2}{2} \right)\ud t+h_t \eta\, \ud W_t, \quad Z_0>0.\quad{}
\end{align}

%Suppose we are given a concave, increasing and twice continuously differentiable utility function $U:\mathbb{R}_+\to\mathbb{R}$. Optimization problem of the trader is given  by
%\begin{align}\max\mathbb{E}^{z,s,i}[U(Z_T)],\end{align}
%where $\mathbb{E}^{z,s,i}$ denotes the conditional expectation given $Z_t=z$, $S_t=s$ and $Y_t=i$.
In what follows, we study the optimization problem for a trader who is endowed with a logarithmic utility in case of regime switching and risk penalization. First, we consider the situation where the trader may observe the Markov chain Y that
influences the dynamics of price processes and the spread. Subsequently, we assume that the Markov chain is not observable and solve the optimization problem under partial information.

\section{Optimization problem under full information}\label{sec:full_info}
\noindent In this section, we suppose that the trader can observe all sources of randomness in the market. Her penalized wealth at time $T$ is given by
\begin{align}
\nonumber Z_T=z\exp\left\{ \int_t^T \left(h_u\left(\kappa (\theta(Y_u)-S_u)-\frac{\eta^2}{2}+\rho\sigma\eta \right)-\frac{h_u^2\eta^2(1+\varepsilon)}{2}+r\right)\,\ud u \right.\\
\left.+\int_t^T h_u\eta\, \ud {W}_u \right\}, \end{align}
for every $h \in \mathcal{A}$. Note that, condition \eqref{eq:A2} guarantees that the stochastic integral in the above expression is a true martingale and hence has a zero expected value.

Formally the trader faces the following optimization problem
\begin{align}\max\,\,\mathbb{E}^{t,z,s,i}[\log Z_T], \label{eq:objective}\end{align}
where $\mathbb{E}^{t,z,s,i}$ denotes the conditional expectation given $Z_t=z$, $S_t=s$ and $Y_t=e_i$.
We define the value function of the trader by
\begin{align}\label{vf}V(t,z,s,i):=\underset{h\in\mathcal{A}}\sup\,\,\mathbb{E}^{t,z,s,i}\left[\log{Z_T}\right].\end{align}
From now on, we use the following notation for the partial derivatives: for every function $g:[0,T]\times\mathbb{R}_{+}\times\mathbb{R}\to\mathbb{R},$ we write, for instance, $\frac{\partial g}{\partial t}=g_t.$
%We also define $b(s,i)={\kappa(\theta_{i}-s)-\frac{\eta^2}{2}+\rho\sigma\eta}$.

In the following theorem we characterize the optimal strategy and the corresponding value function.
\begin{theorem}\label{optimalfull} Consider a trader with a logarithmic utility function with risk penalization parameter $\varepsilon\geq0$. Then the optimal portfolio strategy $h^*\in \mathcal{A}$ is
\begin{align}h^{\ast}(t,s,i)=\frac{1}{1+\varepsilon}\left(\frac{\kappa\left(\theta_i-s\right)}{\eta^2}+\frac{\rho \sigma}{\eta}-\frac{1}{2}\right).\end{align}
The value function is of the form \begin{align}V(t, z, s, i )=\log(z)+r(T-t)+d(t) s^2+ c(t,i) s +f(t, i),\end{align}
where the function $d(t)$ is given by
\begin{align}
d(t)=\frac{\kappa}{4\eta^2(1+\varepsilon)}\left(1-e^{-2\kappa(T-t)}\right),
\end{align}
and the functions $c(t,i)$ and $f(t,i)$ for $i\in\{1, \dots, K\}$ solve the following system of ordinary differential equations
\begin{align}
&c_t(t,i)- \kappa c(t,i)+2\kappa \theta_i d(t)- \frac{\kappa^2 \theta_i-\kappa \frac{\eta^2}{2}+\kappa\rho\sigma\eta}{\eta^2 (1+\varepsilon)}+ \sum_{j=1}^K c(t,j)q^{ij}=0,\label{eq:system1}\\
&f_t(t,i)+\! d(t)\eta^2 + \kappa \theta_i c(t,i)+ \frac{\left(\kappa \theta_i -\frac{1}{2}\eta^2+\rho\sigma\eta\right)^2}{2\eta^2(1+\varepsilon)}\!+\!\sum_{j=1}^K f(t,j) q^{ij}=0    \quad{}\label{eq:system}
\end{align}
with terminal conditions $c(T,i)=0$ and $f(T,i)=0$ for all $i\in \{1,\dots,K \}$.
\end{theorem}

\begin{proof}
We first apply pointwise optimization  to obtain the optimal portfolio strategy. %To this, by It\^{o}'s lemma we can explicitly characterize the portfolio value at time $T$ as
By computing the expectation in \eqref{eq:objective}, we get
\begin{align}\nonumber\mathbb{E}^{t,z,s,i}[\log Z_T]=&\log(z)+r(T-t)-\mathbb{E}^{t,s,i}\left[ \int_t^T \frac{h_u^2\eta^2(1+\varepsilon)}{2}\,\ud u\right]     \\
\label{logustint}&+\mathbb{E}^{t,s,i}\left[ \int_t^T h_u\left(\kappa (\theta(Y_u)-S_u)-\frac{\eta^2}{2}+\rho\sigma\eta \right)\ud u\right],\end{align}
where, according to the previous notation, $\mathbb{E}^{t,s,i}$ denotes the conditional expectation given $S_t=s$ and $Y_t=e_i$.
The first order condition given by
\begin{equation}
-h_t^{*}\eta^2(1+\varepsilon)+\kappa(\theta(Y_t)-S_t)-\frac{\eta^2}{2}+\rho\sigma\eta=0,
\end{equation}
provides the following candidate for the optimal strategy
\begin{align}\ds h^{\ast}(t,s,i)=\frac{1}{1+\varepsilon}\left(\frac{\kappa\left(\theta_i-s\right)}{\eta^2}+\frac{\rho \sigma}{\eta}-\frac{1}{2}\right).\end{align}
The second order condition, $-\eta^2 (1+\varepsilon)<0$, ensures that $h^*$ is the well defined maximizer and hence the optimal portfolio strategy.
By inserting the optimal strategy into \eqref{logustint}, we get a stochastic representation for the optimal value, that is,
\begin{align}\label{stochrepvaluefull}
\log(z)+r(T-t)+\mathbb{E}^{t,s,i}\left[ \int_t^T \frac{(\kappa (\theta(Y_u)-S_u)-\frac{\eta^2}{2}+\rho\sigma\eta )^2}{2\eta^2(1+\varepsilon)}\,\ud u\right].
\end{align}
Next, we characterize the value function by means of Feynman-Kac formula for Markov-modulated diffusion processes, see \cite{Baran2013} and \cite{escobar2015portfolio}. To this, for every $i\in{1, \dots, K}$ we define functions $u(\cdot,\cdot, i):[0,T]\times \R \to \R_+$ by
\begin{align}
u(t,s,i)= \mathbb{E}^{t,s,i}\left[ \int_t^T \frac{(\kappa (\theta(Y_u)-S_u)-\frac{\eta^2}{2}+\rho\sigma\eta )^2}{2\eta^2(1+\varepsilon)}\,\ud u\right].
\end{align}
Then for every $i\in{1, \dots, K}$, functions $u(\cdot,\cdot,i)$, satisfy
\begin{align}
&\nonumber u_t(t,s,i)+\kappa(\theta_i-s)u_s(t,s,i)+\frac{\eta^2}{2}u_{ss}(t,s,i)\\
\label{refpde}&\quad+\sum_{j=1}^Ku(t,s,j) q^{ij}+ \frac{(\kappa(\theta_i-s) +\frac{\eta^2}{2}+\rho\sigma\eta)^2}{2\eta^2(1+\varepsilon)}=0,
\end{align}
with the terminal condition $u(T,s,i)=0$.
Suppose that function $u(t,s,i)$ is of the form $u(t,s,i)=d(t)s^2+c(t,i)s+f(t,i)$.
By using this ansatz, we get the following equation
\begin{align}
\nonumber 0=&c_t(t,i)s+d_t(t)s^2+f_t(t,i)+ \eta^2d(t)  + \frac{(\kappa(\theta_i-s)-\frac{\eta^2}{2}+\rho\sigma\eta)^2}{2\eta^2 (1+\varepsilon)} \\
& +\kappa(\theta_i-s)(c(t,i)+2d(t)s) +\sum_{j=1}^K ( c(t,j)s+f(t,j)) q^{ij}.
\end{align}
%where $\alpha_i$ and $\beta_i$ are defined in the statement of the proposition.
Collecting together the terms with $s^2$, $s$ and the remaining ones we get that the function $d(t)$ solves
\begin{align}
&d_t(t)-2\kappa d(t)+ \frac{\kappa^2}{2\eta^2(1+\varepsilon)}=0,\quad d(T)=0,
%&c_t(t,i)-\kappa c(t,i)+\frac{\kappa(-\frac{1}{2}\eta^2+\rho\sigma\eta)}{\eta^2(1+\varepsilon)}+2\alpha_id(t)=0,
\end{align}
and for every $i\in \{1,\dots,K \}$, $c(t,i)$  and $f(t,i)$ solve the system of ODEs in \eqref{eq:system1} and \eqref{eq:system}, respectively, see e.g., Theorem 3.9 in \citet{teschl2012ordinary}.
\end{proof}

\begin{remark}
\begin{itemize}
\item[i)]Note that the optimal value is always positive provided that $z>1$, and the expectation in \eqref{stochrepvaluefull} can also be evaluated by computing the first and second moments of the Markov-modulated Ornstein--Uhlenbeck process. This can be achieved, for example as given in \citet{huang2016markov}, by solving a non-homogeneous linear system of differential equations.
\item[ii)] In the current setting the market is in general incomplete implying that, for instance, we can not rely on the martingale approach, see, for example, \citet{bjork2010optimal}.
\end{itemize}
\end{remark}

%and denote by $\mathcal L^h$ the generator of the process $(t,Z, S, Y)$, that is

The optimal portfolio strategy $h^{\ast}$ has three components. The component related to dollar-neutrality is given by $\frac{1}{2(1+\varepsilon)}$. This is intuitively clear considering ``non-pairs'' in the sense that there is no correlation ($\rho=0$) and no-cointegration ($\kappa=0$). The other two components are arising from the dependence structure between two stocks. Namely, the first component $\frac{\kappa\left(\theta_i-s\right)}{(1+\varepsilon)\eta^2}$ is related to the co-integration between two stocks, whereas the second component $\frac{\rho \sigma}{\eta({1+\varepsilon})}$ is related to the correlation structure. To wit, suppose now that the current spread is equal to the long-term mean of the current regime, that is $(\theta_i-s)=0$ or $\kappa=0$, then the optimal strategy for a given $\varepsilon>0$ is determined by only the correlation $\rho$ between first stock and spread scaled by the ratio of volatilities of both. One can interpret this case as the dollar-neutral investment strategy in assets with correlated returns. On the other hand, if $\rho$ is zero, the optimal strategy is only determined by the spread dynamics.

\begin{remark}
Suppose that, instead of a \textit{dollar-neutral} strategy, the trader wants to use a \textit{beta-neutral} strategy, that is a strategy of the form $\beta_1h^{(1)}+\beta_2h^{(2)}=0$, where $\beta_1$ and $\beta_2$ denote CAPM betas of $S^{(1)}$ and $S^{(2)}$, respectively. Then the optimal strategy is given by
\begin{align}h^{\ast}(t,s,i)=\frac{1}{1+\varepsilon}\left(\frac{\mu_i\beta_2(\beta_2-\beta_1)+\beta_1\beta_2\kappa\left(\theta_i-s\right)-\beta_1 \beta_2\frac{\eta^2}{2}+\beta_1\beta_2\rho \sigma\eta}{(\sigma(\beta_2-\beta_1)-\beta_1\eta)^2}\right),\end{align}
and the value function has the similar structure as in the dollar-neutral case given above.
\end{remark}

\section{Optimization Problem under Partial Information }\label{sec:partial_info}
\noindent We assume now that the state process $Y$ is not directly observable by the trader. Instead, she observes the price processes $S^{(1)}$  and $S^{(2)}$  and she knows the model parameters. Hence, information available to the trader is carried by the natural filtration of $S^{(1)}$ and $S^{(2)}$. This is equivalent to the set of information carried by $S^{(1)}$ and the spread $S$, that is,
\begin{align}\mathbb{F}=(\mathcal{F}_t)_{t\geq 0}, \ \mathcal{F}_t=\sigma\{ S_u, S^{(1)}_u, 0 \le u \le t\}, \ \mathcal{F}_t\subset \mathcal{G}_t.\end{align}
In the sequel we assume that filtration $\bF$ satisfies the usual hypotheses.

\paragraph{\bf Admissible Investment Strategies.} Decisions of the trader should depend only on the information available to her at time $t$. That is, we consider self-financing investment strategies such that $h$ is $\bF$-progressive. Then we have the following definition of admissible strategies under partial information.
\begin{dfn}
An $\bF$-progressive self-financing investment strategy $h$ that satisfies \eqref{eq:A1} and \eqref{eq:A2} is an $\bF$-\emph{admissible investment strategy}. We denote the set of $\bF$-admissible strategies by $\mathcal{A}^\bF$.
\end{dfn}

The partially informed trader aims to maximize the expected utility $\mathbb{E}[\log Z_T],$ over the class $\mathcal A^\bF$. In this case, we naturally end up with an optimal control problem under partial information. In the next part, to solve such a problem we will derive an equivalent control
problem under full information via the so-called reduction approach, see, e.g., \citet{fleming1982optimal}. This requires the derivation of the filtering equation for the
unobservable state variable. After reduction, the corresponding control problem can be interpreted as one with smooth transitions governed by the  dynamics of filtered probabilities. We discuss this aspect in Section \ref{sec:2stateMC} for the case of a two-state Markov chain.

\subsection{The filtering equation}
In this section we address the problem of characterizing the conditional distribution of the unobservable Markov chain $Y$, given the observation. In our setting, the observations process is given by the pair
\begin{align}\left(\ud R_t,\ud S_t\right)^{\top}=A(t,Y_t,S_t )\ud t+\Sigma \ud B_t,\end{align}
where process $R$ is the log-return of $S^{(1)}$, i.e., $\ud R_t=\frac{\ud S^{(1)}_t}{S^{(1)}_t}$ with $R_0=0$,  $B=(W^{(1)}, W^{(2)})^\top$ is a $2$-dimensional $\mathbb{G}$-Brownian motion independent of $Y$ and
\begin{gather}
A(t,Y_t, S_t)=  \left( \begin{array}{c}
\mu(Y_t)  \\
\kappa(\theta(Y_t)-S_t)   \end{array} \right) ,
\quad  \Sigma= \left( \begin{array}{cc}
    \sigma & 0  \\
   \rho\eta & \sqrt{1-\rho^2}\eta  \end{array} \right), \quad t \in [0,T].
\end{gather}
Note that process $R$ and $S^{(1)}$ generate the same information.

For any function $f$, we denote by $\widehat{f(Y)}$ the optional projection with respect to filtration $\bF$, that is $\widehat{f(Y_t)}=\esp{f(Y_t)|\F_t}$, a.s., for every $t \in [0,T]$. Process $\widehat{f(Y)}$, for every function $f$, provides the filter. By the finite state property of the Markov chain we get that
\begin{align}
\widehat{f(Y_t)}=\sum_{j=1}^K f(e_j)p^j_t, \quad t \in [0,T],
\end{align}
where $p^j_t=\P(Y_t=e_j|\F_t)$, $t \in [0,T]$. Then, in order to characterize the conditional distribution of $Y$, it is sufficient to derive the dynamics of  the processes $p^j$, $j\in\{1,\dots,K\}$.
To this,  we will use the so-called \emph{innovations approach}. This method is based on finding a suitable $\bF$-progressive process that drives the dynamics of the filter, see e.g., \citet{wonham1964} and \citet{elliott1994hidden} for more details.
We define the 2-dimensional process $I=(I^{(1)}, I^{(2)})^\top$ by
\begin{align}\label{What}
I_t=B_t+\int_0^t\Sigma^{-1}(A(u, Y_u, S_t)-\widehat{A(u, Y_u, S_t)})\,\ud u, \quad t \in [0,T].
\end{align}
Explicitly we have
\begin{align}
I^{(1)}_t&=W^{(1)}_t+\int_0^t\frac{\mu(Y_u)-\widehat{\mu(Y_u)}}{\sigma}\,\ud u,\\
I^{(2)}_t&=W^{(2)}_t+\int_0^t\frac{\sigma \kappa (\theta(Y_u)-\widehat{\theta(Y_t)})-\rho\eta(\mu(Y_u)-\widehat{\mu(Y_u)})}{\sigma \eta \sqrt{1-\rho^2}}\,\ud u,
\end{align}
for every $t \in [0,T]$.
\begin{remark}\label{inn} The process $I$ is called innovation process and it is well known that $I$ is an  $(\bF, \P)$-Brownian motion, see Proposition 2.30 in \citet{bain2009fundamentals}.
\end{remark}

Note that, since the signal $Y$ and the Brownian motion $B$ driving the observation process are assumed to be independent, the filtration $\bF$ coincides with the natural filtration of the innovation process, see Theorem 1 in \citet{allinger1981new}. Then, by Theorem III.4.34-(a) in \citet{jacod2013limit} every $(\P, \bF)$-local martingale $M$ admits the following representation:
\begin{align}\label{eq:mg_representation}
M_t=M_0+\int_0^tH_u\, \ud I_u,\quad t\in[0,T],
\end{align}
for some $\bF$-predictable 2-dimensional process $H$ such that
\begin{equation}
\int_0^T\|H_u\|^2\,\ud u<\infty \quad \P-\text{a.s.}
\end{equation}

We recall the notation $\bmu= (\mu_1,\dots,\mu_K)^{\top}$, where $\mu_i=\mu(e_i)\in \mathbb{R}$, and $\btheta= (\theta_1,\dots,\theta_K)^{\top}$, where $\theta_i=\theta(e_i)\in \mathbb{R}$. Also introduce $\mathbf{f}=(f_1,\dots,f_K)^{\top}$, where $f_i=f(e_i)\in \mathbb{R}$. The next theorem provides the filter dynamics.
\begin{theorem}\label{thm:filter}
For every $i\in\{1,\dots, K\}$, the filter process $p^i$ satisfies
\begin{align}
\nonumber p^i_t&=p^i_0+\int_0^t\sum_{j=1}^K q^{ji}p^j_u\, \ud u+ \frac{1}{\sigma}\int_0^tp_u^i(\mu^i-\bmu^{\top}p_u)\,\ud I^{(1)}_u \\
&\quad{}+ \frac{1}{\sigma \eta \sqrt{1-\rho^2}}\int_0^t p^i_u\left( \sigma \kappa  \left(\theta_i-\btheta^{\top}p_u \right)-\eta\rho\left(\mu_i-\bmu^{\top}p_u \right)\right)\, \ud I^{(2)}_u,\quad{}p^i_0=\Pi^i,\label{eq:filtering-1}
\end{align}
for every $t \in [0,T].$
\end{theorem}

\begin{proof}
Consider the semimartingale decomposition of $f(Y)$ given by
\begin{align}
f(Y_t)= f(Y_0)+\int_0^t \langle Q\mathbf f, Y_{u^-}\rangle\, \ud u + M^{(1)}_t,\quad t \in [0,T],
\end{align}
where $M^{(1)}$ is a $(\mathbb G, \P)$-martingale. Now, projecting over $\bF$ leads to
\begin{align}
\widehat{f(Y_t)}- \widehat{f(Y_0)}-\int_0^t \langle Q\mathbf f, \widehat{Y}_{u^-}\rangle\, \ud u = M^{(2)}_t,\quad t \in [0,T],
\end{align}
where $M^{(2)}$ is an $(\bF, \P)$-martingale. Using the martingale representation in \eqref{eq:mg_representation} we get
\begin{align}
\widehat{f(Y_t)}- \widehat{f(Y_0)}-\int_0^t \langle Q\mathbf f, \widehat{Y}_{u^-}\rangle\, \ud u = \int_0^tH_u\, \ud I_u,\quad t \in [0,T].
\end{align}
Let $m_t=I_t+\int_0^t\Sigma^{-1} \widehat{A(u, Y_u)}\,\ud u$, for every $t \in [0,T]$. Computing the product $f(Y)\cdot m$ and projecting on $\bF$, we obtain
\begin{align}\label{eq:filt1}
\widehat{f(Y_t)\cdot m_t}= \int_0^t\!\!\! m_u \langle Q\mathbf f, \widehat{Y}_u\rangle \, \ud u + \int_0^t\!\! \Sigma^{-1}\widehat{f(Y_u) A(u, Y_u)}\,\ud u + M^{(3)}_t,\,\, t \in [0,T],
\quad{}\end{align}
for some $(\bF, \P)$-martingale $M^{(3)}$.
We now compute the product  $\widehat{f(Y)}\cdot m$ as
\begin{align}\label{eq:filt2}
\widehat{f(Y_t)} \cdot m_t= \int_0^t \!\!m_u \langle Q\mathbf f, \widehat{Y}_u\rangle\,\ud u + \int_0^t\!\! \Sigma^{-1}\widehat{f(Y_u)} \widehat{A(u, Y_u)}\, \ud u + \int_0^t\!\! H_u \,\ud u + M^{(4)}_t.\qquad{}
\end{align}
for every $t \in [0,T]$, where $M^{(4)}$ is an $(\bF, \P)$-martingale.
Comparing the finite variation terms in \eqref{eq:filt1} and \eqref{eq:filt2}, we get
\begin{align}
H^{(1)}_t&=\frac{\widehat{f(Y_t) \mu(Y_t)}-\widehat{f(Y_t)}\widehat{\mu(Y_t)}}{\sigma},\\
H^{(2)}_t&=\frac{\sigma \kappa (\widehat{f(Y_t)\theta(Y_t)}-\widehat{f(Y_t)}\widehat{\theta(Y_t)})-\eta\rho(\widehat{f(Y_t)\mu(Y_t)}-\widehat{f(Y_t)}\widehat{\mu(Y_t)})}{\sigma\eta\sqrt{1-\rho^2}},
\end{align}
for every $t \in [0,T]$. Finally choosing $f(Y_t)=\I_{\{Y_t=e_i\}}$, we obtain the result.
\end{proof}

%The proof of Theorem \ref{thm:filter} is postponed to Appendix.

\begin{remark} Here notice that the drift and diffusion coefficients in \eqref{eq:filtering-1} are continuous, bounded and locally Lipschitz. This implies that $p$ is the unique strong solution of the filtering equation \eqref{eq:filtering-1}.
\end{remark}

\subsection{Reduction of the Optimal Control Problem}
The semimartingale decomposition  of $Z$ and $S$ with respect to the observation filtration are given by
\begin{align}
\nonumber Z_t=&Z_0+\int_0^tZ_u\left(h_u\left(\kappa(\btheta^{\top}p_u-S_u)-\frac{\eta^2}{2}+\rho\sigma\eta \right)+r-\frac{\varepsilon \eta^2h_u^2}{2}  \right)\ud u\\
&+\eta \int_0^t h_uZ_u\left( \rho\, \ud I^{(1)}_u+\sqrt{1-\rho^2}\, \ud I^{(2)}_u\right) ,\quad t \in [0,T],\label{eq:semimg_Z}
\end{align}
and
\begin{align}\label{shat}
S_t=S_0+\int_0^t\kappa(\btheta^{\top}p_u-S_u)\,\ud u+\eta\int_0^t\left( \rho\,\ud I^{(1)}_u+\sqrt{1-\rho^2}\,\ud I^{(2)}_u\right), \quad t \in [0,T] .
\end{align}

Thanks to uniqueness of the solution of the filtering equation we can consider the $(K+2)$-dimensional process $(Z,S,p)$ as the state process and introduce the equivalent optimal control problem under full information, see, e.g., \citet{fleming1982optimal}. We have
\begin{align}\label{opt2}
\max\,\,\mathbb{E}^{t,z,s, \p}[\log{Z_T}],\end{align}
where $\mathbb{E}^{t,z,s,\p}$ denotes the conditional expectation given $Z_t=z$, $S_t=s$ and $p_t=\p$, where $(z,s,\p)\in\mathbb R_+\times \mathbb R \times \Delta_K$, with $\Delta_K$ denoting the $(K-1)$-dimensional simplex.
We define the value function of the trader by
\begin{align}\label{vf}V(t,z,s,\p):=\underset{h\in\mathcal{A}^\bF}\sup\,\,\mathbb{E}^{t,z,s,\p}\left[\log{Z_T}\right].\end{align}
To obtain the optimal strategy it is possible to apply pointwise maximization, which also leads to an explicit characterization for the value function. This is given in the next theorem.

\begin{theorem}\label{optimalpartial} Consider a trader with a logarithmic utility function with risk penalization parameter $\varepsilon\ge0$. Then the optimal portfolio strategy $h^* \in \mathcal{A}^{\bF}$ under partial information is
\begin{align}h^{\ast}(t,s,\p)=\frac{1}{1+\varepsilon}\left(\frac{\kappa\left(\btheta^{\top}\p-s\right)}{\eta^2}+\frac{\rho \sigma}{\eta}-\frac{1}{2}\right).
\end{align}
The value function is of the form \begin{align}V(t, z, s, \p )=\log(z)+r(T-t)+ d(t) s^2+ c(t,\p)s+f(t,\p),\end{align}
where the function $d(t)$ is given by
\begin{align}
d(t)=\frac{\kappa}{4\eta^2(1+\varepsilon)}\left(1-e^{-2\kappa(T-t)}\right),
\end{align}
and the functions $c(t,\p)$ and $f(t,\p)$ solve the following system of partial differential equations:
\begin{align}
&c_t(t,\p)\!+\!\frac{1}{2}\!\!\sum_{i,j=1}^K \!\!\widetilde{\alpha}^{ij}(\p)c_{p^ip^j}(t,\p)\!+\!\!\!\sum_{i,j=1}^K \!\!c_{p^i}(t,\p)q^{ji}p^j \! +\!\kappa\big(2d(t)\btheta^{\top}\p-c(t,\p)\big)\!-\!\gamma(\p)\!=\!0,\ \quad{} \label{eq:pde-1}\\
\nonumber&f_t(t,\p)+\frac{1}{2} \sum_{i,j=1}^K \widetilde{\alpha}^{ij}(\p)f_{p^ip^j}(t,\p)+\sum_{i,j=1}^K f_{p^i}(t,\p)q^{ji}p^j+\eta\sum_{i=1}^K c_{p^i}(t,\p)\widetilde{\beta}^i(\p)\\
&\qquad+c(t,\p)\kappa\btheta^{\top}\p+\eta^2d(t)+\frac{\left(\kappa\btheta^{\top}\p-\frac{1}{2}\eta^2+\rho\sigma\eta\right)^2 }{2\eta^2(1+\varepsilon)}=0,\qquad\label{eq:pde-2}
\end{align}
with terminal conditions $c(T,\p)=0$ and $f(T,\p)=0$ for every $\p\in\Delta_K$, and where
\begin{gather}
\widetilde{ \alpha}^{i,j}(\p)= H^{(i,1)}(\p)H^{(j,1)}(\p)+H^{(i,2)}(\p)H^{(j,2)}(\p),\quad i,j \in \{1,\dots, K\},\\
\widetilde{ \beta}^i(\p)=\rho H^{(i,1)}(\p)+\sqrt{1-\rho^2}H^{(i,2)}(\p),  \quad i \in \{1,\dots, K\},   \\
H^{(i,1)}(\p)=p^{i}\frac{(\mu_i-\bmu^{\top}\p)}{\sigma},\qquad H^{(i,2)}(\p)=p^i\frac{ \sigma \kappa  \left(\theta_i-\btheta^{\top}\p \right)-\eta\rho\left(\mu_i-\bmu^{\top}\p \right)}{\sigma\eta\sqrt{1-\rho^2}},\\
\gamma(\p)=\frac{\kappa}{\eta^2(1+\varepsilon)}\left(\btheta^{\top}\p-\frac{1}{2}\eta^2+\rho\sigma \eta \right).
\end{gather}
\end{theorem}

\begin{proof}
The proof of Theorem \ref{optimalpartial} follows the same lines of that of Theorem \ref{optimalfull} and it is provided, for completeness, in Appendix.
\end{proof}

\paragraph{\textbf{Comments and discussion.}}
By Theorem~\ref{optimalfull} and Theorem~\ref{optimalpartial}, optimal strategies depend on both the correlation between two assets and the mean-reverting spread. Moreover, they do not depend on the risk-free rate $r$ because a priori we restrict ourselves to the dollar-neutral pairs trading strategies. Comparing optimal strategies under full and partial information, we can say that the so-called \emph{certainty equivalence principle} holds, i.e., the optimal portfolio strategy in the latter case can be obtained by replacing the unobservable state variable with its filtered estimate.\footnote{This unorthodox definition of certainty equivalence principle is due to \citet{kuwana1995certainty} and used in literature related to partial information models, see e.g., \citet{bauerle2004portfolio}.}

 The effect of risk-penalization on optimal strategies is to increase the risk-aversion uniformly in a constant proportion that is not dependent on time. It effectively decreases the proportion of wealth invested in pairs and increases the proportion of wealth invested in the risk-free asset. Considering the optimal value functions, in both cases, they are quadratic functions of the current value of the spread. However, in both cases, coefficients (factor loadings) on the quadratic term, $s^2$, depend only on time. This result is worth to mention since it means that the trader does not really consider the effect of the partial information on the quadratic level of the current spread. Finally, note that similar results hold true for beta-neutral strategies.

\section{Toy Example: Two-State Markov Chain }\label{sec:2stateMC}
\noindent In this section, we give a toy example of our proposed model, where the unobservable Markov chain has only two states. Here our main aim is to demonstrate certain qualitative features of the model that are difficult to verify analytically. During our analysis, we set $z=1$, $\theta_1=0.1$, $\theta_2=0.6$, $\mu_1=0.2$ and $\mu_2=1$. In the first step, we consider the full information case, where the trader knows the state of the Markov chain. Then, we investigate the case with the partial information.

\subsection{The full information case}
In this part, we employ Theorem~\ref{optimalfull} where we solve the corresponding system of ODEs numerically. In the following, since we set $z=1$ we suppress the dependence of the value function on $z$ and write $V(t,s,i)$ for $i\in\{1,2\}$.

In Figure~\ref{fig:crossing}, we  illustrate optimal values with respect to time to maturity for a given initial state and for different values of initial spread ($s=0.1$, $s=0.3$ and $s=0.7$).  It suggests that for all initial states and for all values of the initial spread, the optimal value increases in time to maturity since trading possibilities increase as there would be more time to trade. Moreover, as it is expected from a pairs trading strategy, the wider the gap between the initial spread and the long-run mean of the initial state's spread, the higher the optimal value provided that there is enough time to have the spread close with high probability.
For example, in Figure~\ref{fig:crossing} (left panel), we observe $V(t,0.1,2)>V(t,0.1,1)$ for all $t$ . This corresponds to the case where the trader could exploit the wide enough gap between initial spread, $s=0.1$, and the long-run mean of the second state, $\theta_2=0.6$. A similar behaviour is observed in Figure~\ref{fig:crossing} (right panel), where in this case the gap between the initial spread, $s=0.7$ and the long-run mean of the first state, $\theta_1=0.1$, is large enough for $V(t,0.7,1)>V(t,0.7,2)$ for all $t$.

However in Figure~\ref{fig:crossing} (middle panel), there is no clear dominance between optimal values corresponding to different initial states. This can be explained by the following observation. The initial spread, which is 0.3, is approximately at the same distance to both states' long-run means hence the intersection point of the two functions $V(\cdot,0.3,1)$ and $V(\cdot,0.3,2)$ depends more on the transition intensities of the Markov chain $q^{12}$ and $q^{21}$. In particular, for this example, fixing all other parameters, the intersection point moves to the right as $q^{12}$ gets larger. Overall we can conclude that the main determinants of the observed dominance are the gap between initial spread and the long-run mean of states, transition probabilities as well as remaining time to maturity.
\begin{figure}[htbp]
\centering
\includegraphics[width=4.35cm,height=6cm]{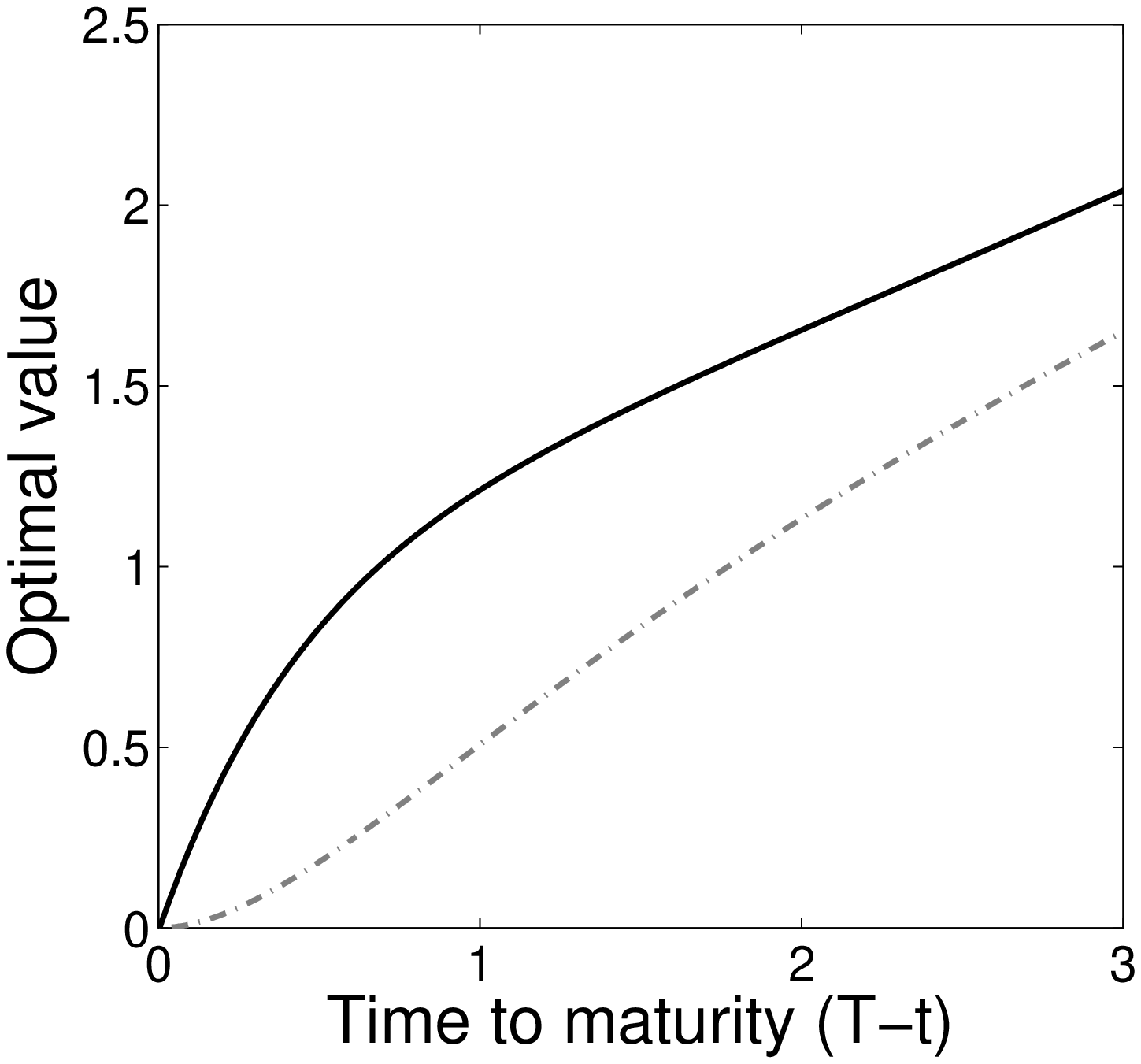}
\hspace{-.6cm}
\includegraphics[width=4.35cm,height=6cm]{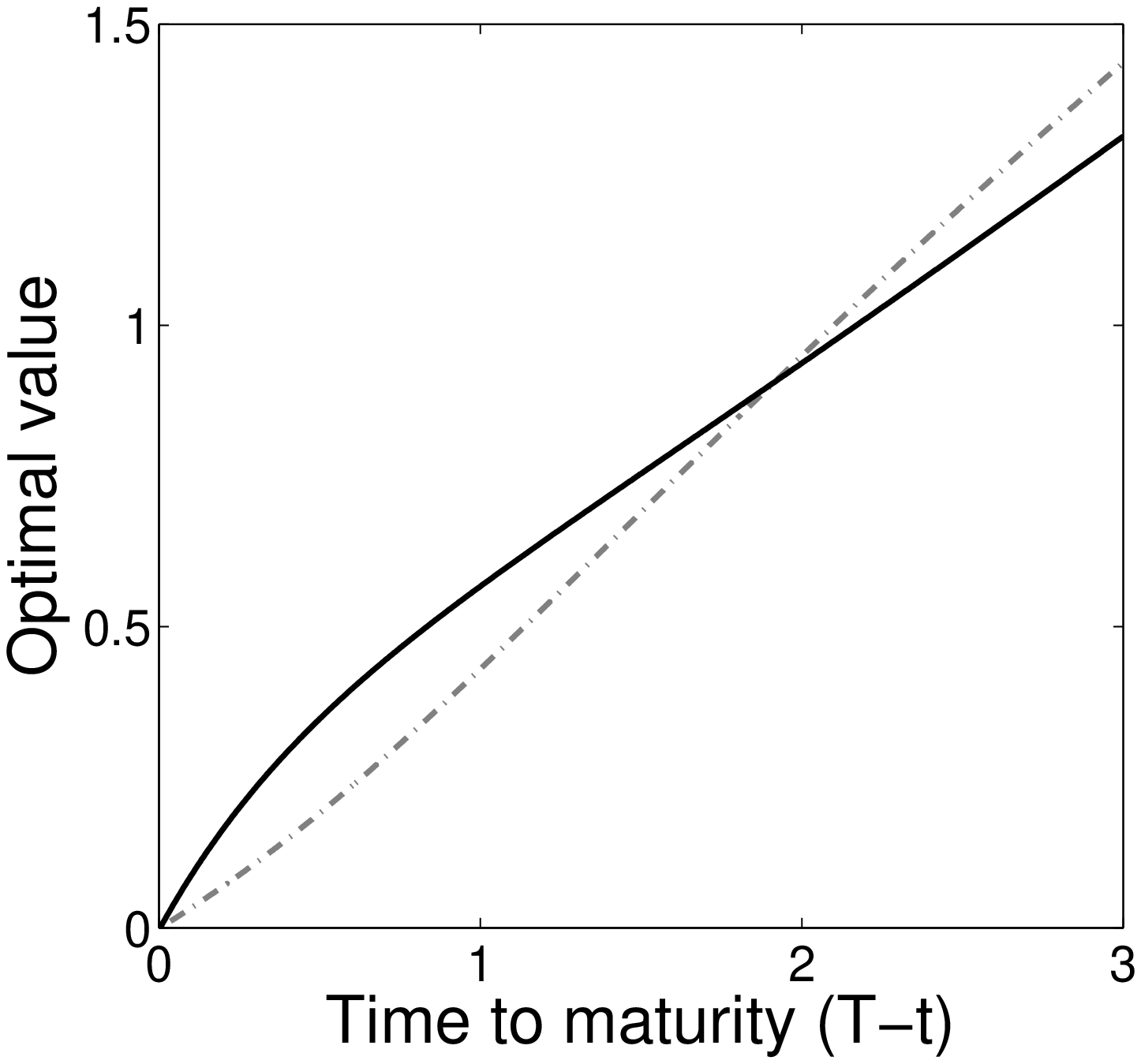}
\hspace{-.6cm}
\includegraphics[width=4.35cm,height=6cm]{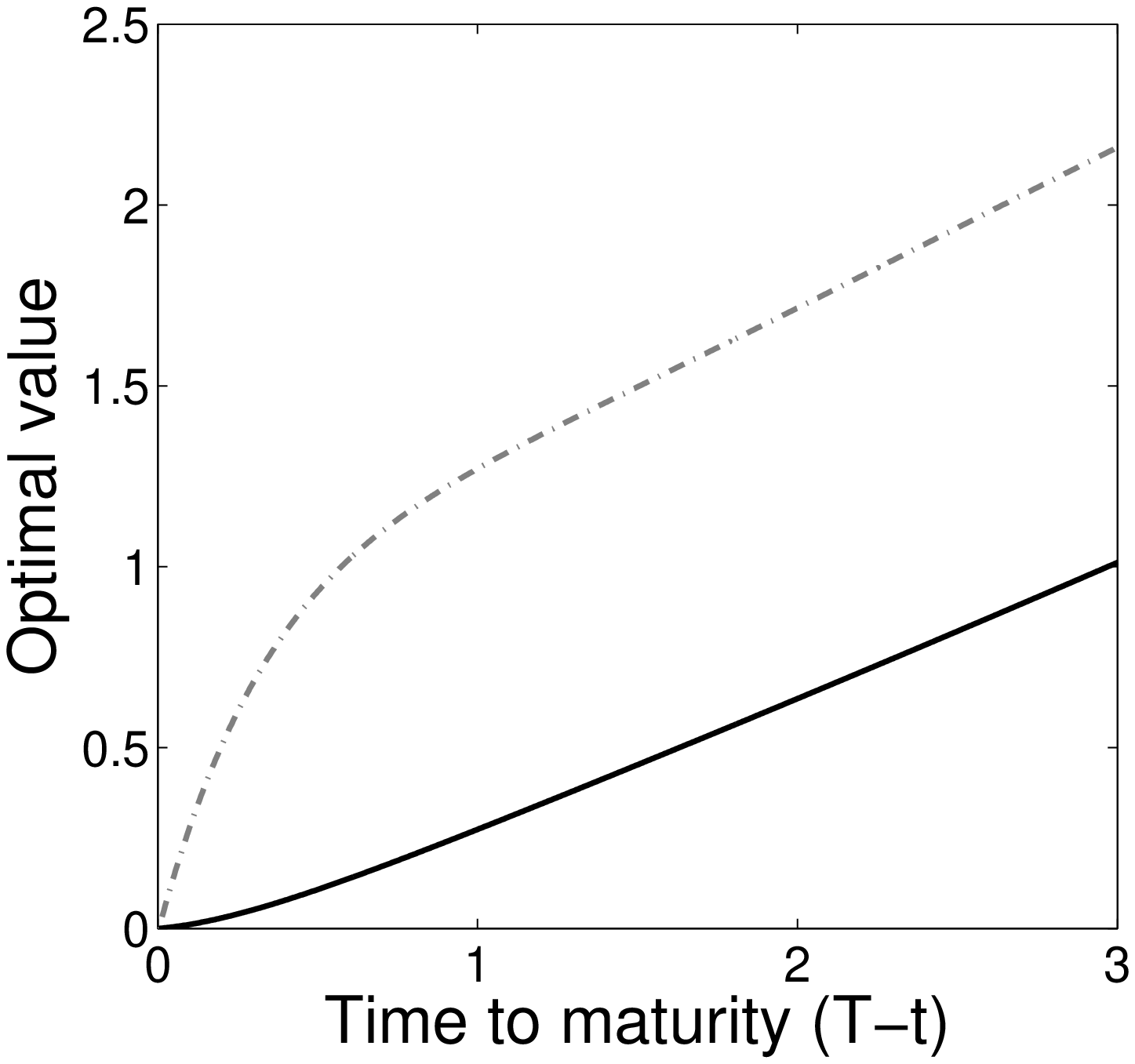}
\caption{Optimal value as a function of time to maturity for different values of initial spread, $s$, when the initial state is $e_1$ (dashed line) or $e_2$ (solid line). Left panel $s=0.1$. Middle panel $s=0.3$. Right panel $s=0.7$. Other parameters:  $z=1$, $r=0.01$, $\theta_1=0.1$, $\theta_2=0.6$, $\kappa=1$, $\rho=0.9$, $\sigma=0.2$, $\eta=0.2$, $\varepsilon=0.3$, $q^{12}=0.7$ and $q^{21}=0.2$.}
\label{fig:crossing}
\end{figure}

Next in Figure~\ref{fig:mmod_avg}, we compare the value of the current optimal portfolio problem with the optimal value computed using the {\em averaged data}.
%In this way we intend to see whether the knowledge of averaged data is sufficient for obtaining the value function.
Let $(\pi,1-\pi)$ denotes the stationary distribution of the Markov chain $Y$. Suppose we have two traders, one of which ignores the Markov modulated nature of the underlying spread and considers the {\em averaged data} $\overline{\theta}=\pi\theta_1+(1-\pi)\theta_2$ as the long-run mean spread. On the other hand, the second trader assumes our proposed Markov modulated model, that is, she acts in line with what Theorem~\ref{optimalfull} suggests.  We want to compare the value function $\overline{V}(t,s)$ obtained in the model  assuming averaged data with the value function in the Markov-modulated case. In this way, we intend to see whether the knowledge of averaged data is sufficient to obtain the optimal value for the current pairs trading problem.
To this, we set $q^{12}=1$ and $q^{21}=2$, and compute $\pi=\frac{q^{21}}{q^{12}+q^{21}}=0.67$. Then, we get $\overline{\theta}=0.27$. In Figure~\ref{fig:mmod_avg} we plot $\overline{V}(t,s)$ versus $\mathbb{E}^{\pi}[V(t,s,Y_t)]=\pi V(t,s,1)+(1-\pi)V(t,s,2).$
We observe that $\mathbb{E}^{\pi}[V(t,s,Y_t)]>\overline{V}(t,s)$. This implies that the averaged data does not contain sufficient information to obtain the optimal value for the pairs trading problem and hence on the average, the second trader performs better than the first one. This result is in contrast with the one for the classical portfolio optimization problem with Markov modulation in the case of logarithmic utility preferences. See Section B of \citet{bauerle2004portfolio}.  We attribute this to the mean-reverting nature of the underlying state variable.

\begin{figure}[htbp]
\centering
\includegraphics[width=6.25cm,height=6cm]{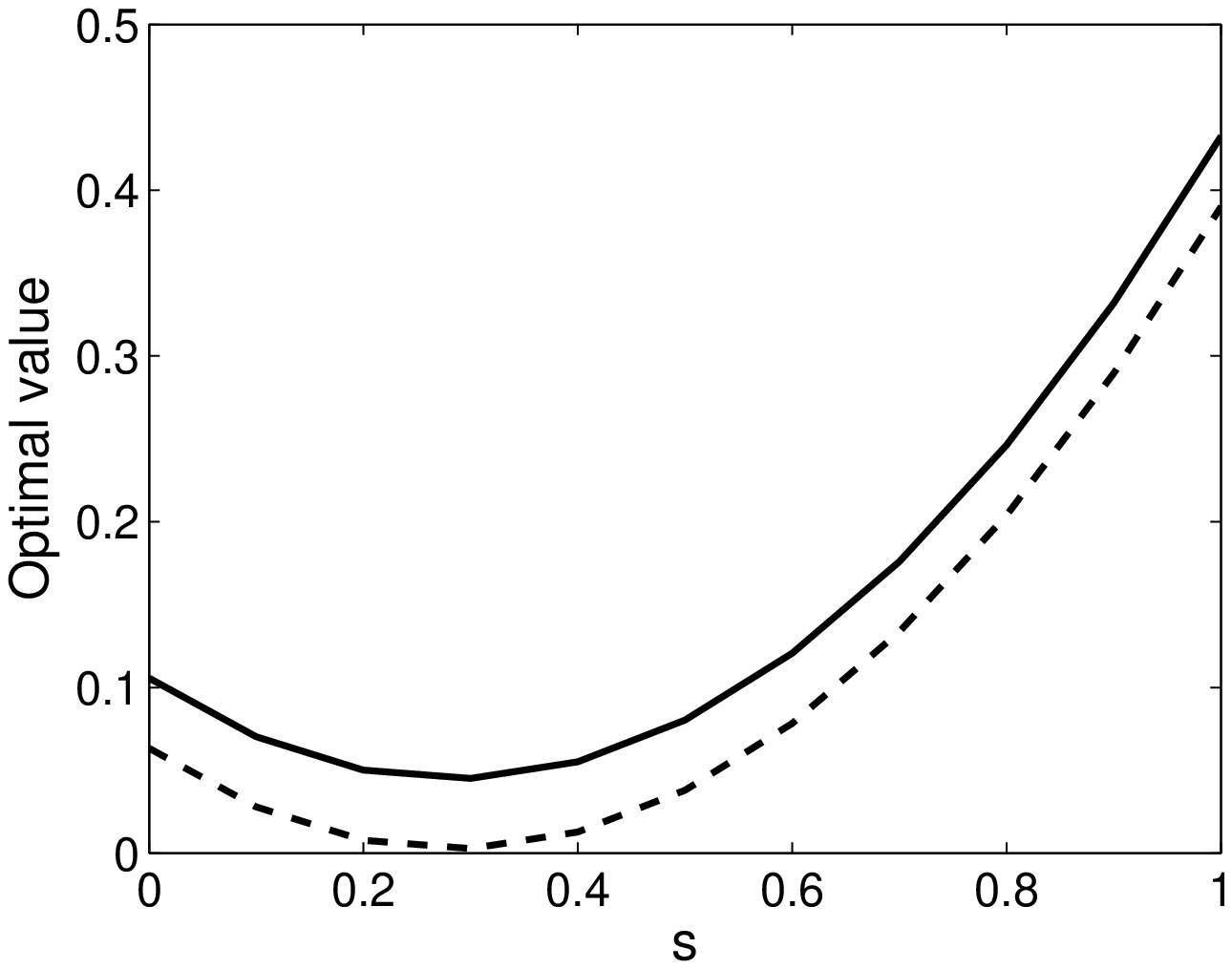}
\includegraphics[width=6.25cm,height=6cm]{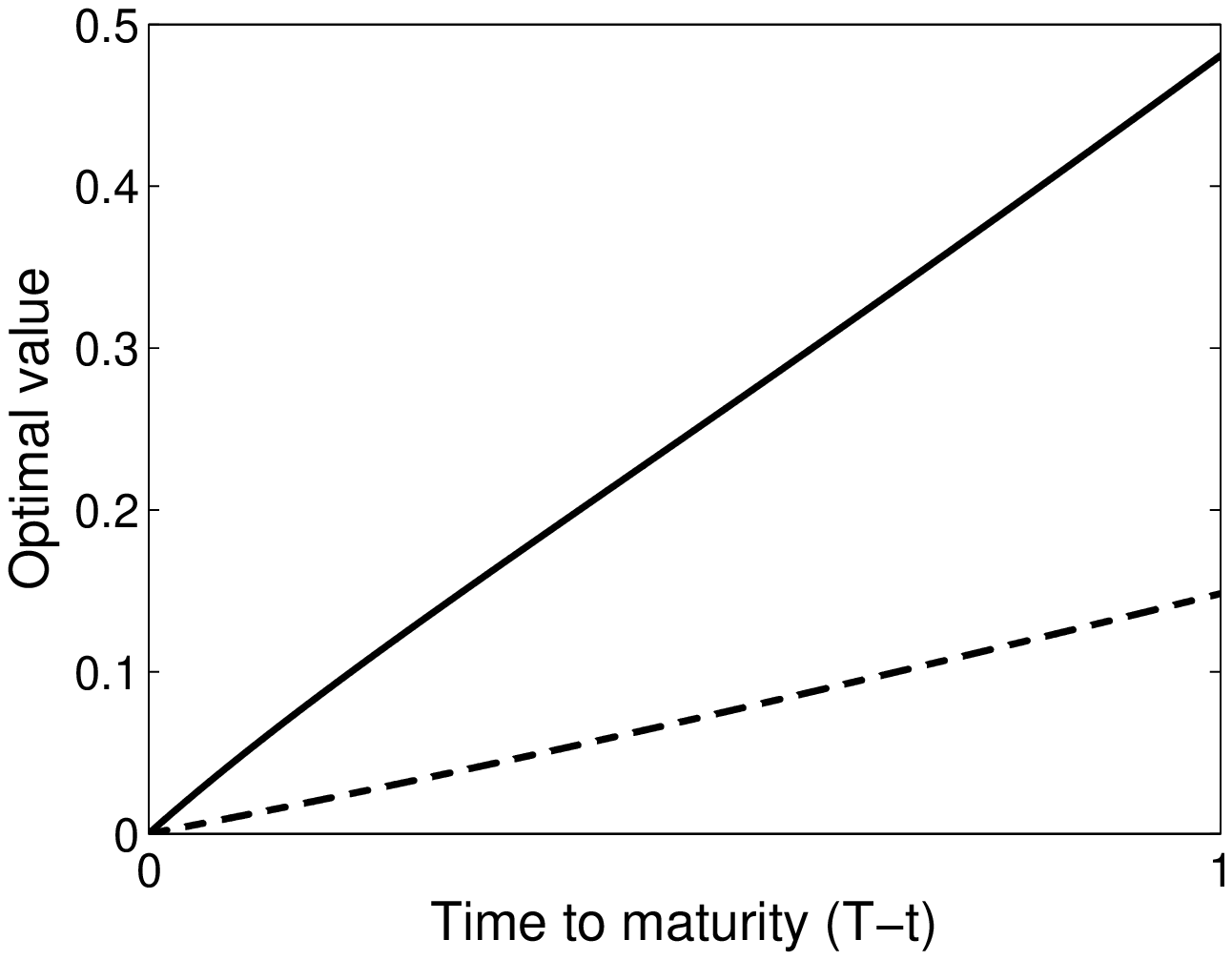}
\caption{Left panel: Optimal value as a function of initial spread for time to maturity $T-t=0.1$ years. Right panel: Optimal value as a function of time to maturity for initial spread $s=0.3$. The solid line (resp. dashed line) indicates the optimal value corresponding to Markov switching case (resp. averaged data case ). Other parameters: $z=1$, $r=0.01$, $\theta_1=0.1$, $\theta_2=0.6$, $\kappa=1$, $\rho=0.9$, $\sigma=0.2$, $\eta=0.2$, $\varepsilon=0.5$, $q^{12}=1$ and $q^{21}=2$.}
\label{fig:mmod_avg}
\end{figure}

Figure~\ref{fig:rho_kappa} depicts the behavior of the value function with respect to the mean-reversion speed $\kappa$, for correlation values $\rho=0.1$ and $\rho=0.9$. In the case without Markov switching one would expect higher values of $\kappa$ to yield higher optimal values since that would imply more visits to the long-run mean generating profit opportunities from pairs trading more frequently. Here, we observe that higher values of $\kappa$ not necessarily lead to larger portfolio values since there is the risk of a regime switch which would result in a sudden change in the long-run mean value.
\begin{figure}[htbp]
\centering
\includegraphics[width=8.25cm,height=6cm]{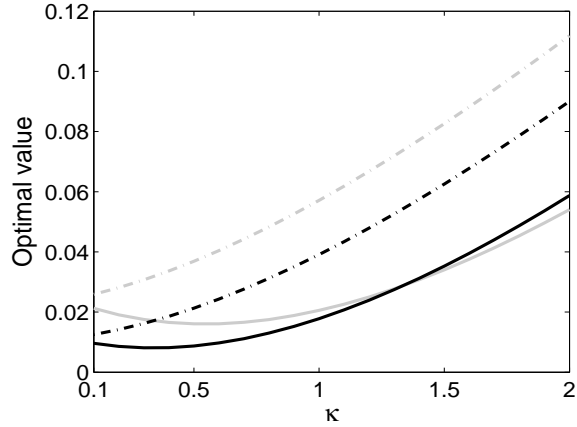}
\caption{Impact of mean reversion speed $\kappa$ on optimal value.  Dashed line (resp. solid line) corresponds to the optimal value when the initial state is $e_1$ (resp. $e_2$). Grey line: $\rho=0.1$, black line: $\rho=0.9$. Other parameters: $T-t=3$ years, $z=1$, $r=0.01$, $\theta_1=0.1$, $\theta_2=0.6$, $s=0.3$, $\eta=0.9$, $\sigma=0.2$, $\varepsilon=0.3$,  $q^{12}=0.7$ and $q^{21}=0.2$.  }
\label{fig:rho_kappa}
\end{figure}

\subsection{The partial information case}
In the partially observable setting, having only two states enables us to reduce the number of state variables for our filtered control problem since $p:=p^{1}=1-p^{2}$. Then we only need the dynamics of $p$, given, after arrangement, by
\begin{align}\label{eq:filtertwostate}
\ud p_t=(q^{12}+q^{21})\left(\frac{q^{21}}{q^{12}+q^{21}}-p_t\right)\,\ud t+\sqrt{\nu^2_1+\nu_2^2}p_t(1-p_t)\,\ud I_t^{(3)},
\end{align}
\noindent
where $\nu_1=\frac{(\mu_1-\mu_2)}{\sigma}$ and $\nu_2=\frac{\sigma\kappa(\theta_1-\theta_2)-\eta\rho(\mu_1-\mu_2)}{\sigma\eta\sqrt{1-\rho^2}}$, and $I^{(3)}=\frac{\nu_1}{\sqrt{\nu_1^2+\nu_2^2}}I^{(1)}+\frac{\nu_2}{\sqrt{\nu_1^2+\nu_2^2}}I^{(2)}$ is an $\mathbb{F}$-Brownian motion. We can write the semimartingale decomposition of wealth and spread processes with respect to filtration $\bF$ as
\begin{align}\label{eq:wealthtwostate}
\nonumber \ud Z_t=&Z_t\left(h_t\left(\kappa(\theta_2+(\theta_1-\theta_2)p_t-S_t)-\frac{\eta^2}{2}+\rho\sigma\eta \right)+r-\frac{\varepsilon \eta^2h_t^2}{2}  \right)\ud t\\
&+\eta h_tZ_t\,\ud \tilde{I}_t,
\end{align}
and
\begin{align}\label{eq:shat}
\ud S_t=\kappa\left(\theta_2+(\theta_1-\theta_2)p_t-S_t\right)\,\ud t+\eta\,\ud \tilde{I}_t,
\end{align}
where $\tilde{I}$ is a $\mathbb{F}$-Brownian motion with $\langle \tilde{I},I^{(3)}\rangle_t=\frac{\nu_1\rho+\nu_2\sqrt{1-\rho^2}}{\sqrt{\nu^2_1+\nu_1^2}}t$.

Note that one can interpret the reduced control problem with state variables $(Z,S,p)$ given by \eqref{eq:wealthtwostate}, \eqref{eq:shat} and \eqref{eq:filtertwostate} as a pairs trading model with smooth transitions. More precisely,one can see $p$ as a state variable process governing smooth transitions between two regimes with different long-term means for the spread, that is, $\theta_1$ and $\theta_2$. The dynamics of $p$ is also very similar to a mean-reverting Jacobi-type (or Wright--Fisher) diffusion used in population genetics to model allele frequencies\footnote{For the Jacobi or Wright--Fisher diffusion, the diffusion coefficient is given by $\sqrt{p(1-p)}.$}, see e.g.,~\citet{ethier1976class}, \citet{sato1976diffusion} or \citet{gourieroux2006multivariate}.

In this case the value function can be written as $V(t,z,s,p^1,p^2)=\widetilde{V}(t,z,s,p)$, and, as in Theorem \ref{optimalpartial}, the
optimal value is given by $\widetilde{V}(t,z,s,p)=\log(z)+r(T-t)+d(t) s^2+\widetilde{c}(t,p)s+\widetilde{f}(t,p)$,
where the function $d(t)$ is given by
\begin{align}
d(t)=\frac{\kappa}{4\eta^2(1+\varepsilon)}\left(1-e^{-2\kappa(T-t)}\right),
\end{align}
and the functions $\widetilde{c}(t,p)$ and $\widetilde{f}(t,p)$ solve the following system of partial differential equations:

\begin{align}
\nonumber &\widetilde{c}_t(t,p)-\frac{\kappa^2(\theta_2+(\theta_1-\theta_2)p)-\kappa(\frac{-\eta^2}{2}+\rho\sigma\eta)}{\eta^2(1+\varepsilon)}-\kappa\widetilde{c}(t,p)+2\kappa(\theta_2+(\theta_1-\theta_2)p)d(t)\\
&+(q^{12}+q^{21})\left(\frac{q^{21}}{q^{12}+q^{21}}-p\right)\widetilde{c}_p(t,p)+\frac{1}{2}(\nu^2_1+\nu_2^2)p^2(1-p)^2\widetilde{c}_{pp}(t,p)=0,\end{align}
\begin{align}
\nonumber &\widetilde{f}_t(t,p)+\frac{\kappa^2(\theta_2+(\theta_1-\theta_2)p)^2+(\rho\sigma\eta-\frac{\eta^2}{2})^2+2\kappa(\theta_2+(\theta_1-\theta_2)p)(\rho\sigma\eta-\frac{\eta^2}{2})}{2\eta^2(1+\varepsilon)}\\
\nonumber &+\eta^2d(t)+\kappa(\theta_2+(\theta_1-\theta_2)p)\widetilde{c}(t,p)+(q^{12}+q^{21})\left(\frac{q^{21}}{q^{12}+q^{21}}-p\right)\widetilde{f}_p(t,p)\\
&+\frac{1}{2}(\nu^2_1+\nu_2^2)p^2(1-p)^2\widetilde{f}_{pp}(t,p)+\kappa(\theta_1-\theta_2)p(1-p)\widetilde{c}_p(t,p)=0,\label{eq:pdeexamp}
\end{align}
with terminal conditions $\widetilde{c}(T,p)=0$ and $\widetilde{f}(T,p)=0$ for every $p\in[0,1],$

We use an explicit finite-difference method to solve the system of PDEs given in \eqref{eq:pdeexamp} numerically. In order to guarantee the positivity of the scheme we use forward-backward approximation for the first order derivatives. The value function in the partial information case has a similar behavior with respect to the parameters as the one in the full information case. However, we stress that in the partial information setting, also the drift parameters $\mu_1$ and $\mu_2$ play a role. In particular, the relative values of $\mu_1$, $\mu_2$ and the noise parameters $\sigma$ and $\eta$ control for the precision of the filtered probability estimates.

In Figure~\ref{fig:gains_filtering} we illustrate that the trader benefits from using filtered estimates instead of average data. As it can be seen clearly, gains from filtering increase in time to maturity. On the other hand, gains get smaller as $p$ moves towards $\frac{1}{2}$, which represent the most uncertain situation.

\begin{figure}[htbp]
\centering
\includegraphics[width=10cm,height=8cm]{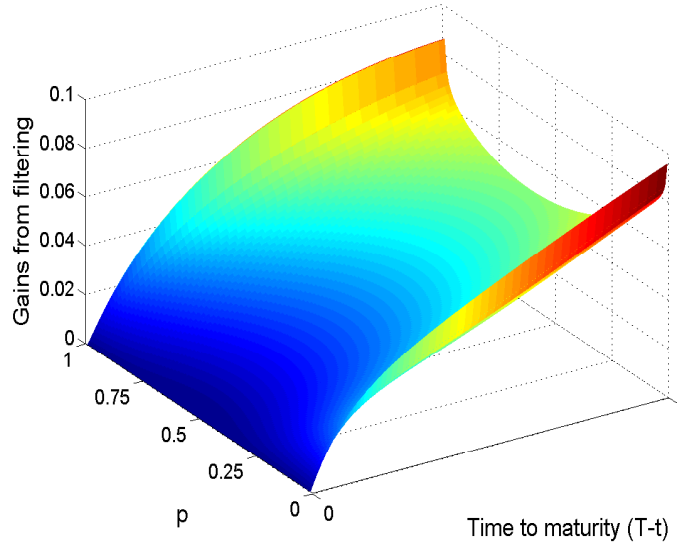}
\caption{Gains from filtering as a function of $p$ and time to maturity. Other parameters: $z=1$, $r=0.01$, $\theta_1=0.1$, $\theta_2=0.6$, $\mu_1=0.2$, $\mu_2=1$, $\kappa=1$, $\rho=0.9$, $\sigma=0.2$, $\eta=0.2$, $\varepsilon=0.5$, $s=0.3$, $q^{12}=1$ and $q^{21}=2$.}
\label{fig:gains_filtering}
\end{figure}
We can summarize the findings of this section as follows:
\begin{itemize} 
\item[(a)] the wider the gap between the initial spread and the long-run
mean of the initial state's spread, the higher the optimal value provided that there is
enough time to have the spread close with high probability, 
\item[(b)] the average data does not contain sufficient information to obtain the optimal value for the current pairs trading problem,
\item[ (c)] higher values of the mean reversion speed $\kappa$ does not necessarily imply higher optimal values,
\item[ (d)] in the partial information setting, there is a gain from filtering due to the convexity originating from using filtered probabilities.
\end{itemize}
\section{Conclusion}\label{sec:discussion}
In this paper, we have considered the pairs trading for a trader with logarithmic utility preferences and risk penalized terminal wealth. By penalizing the terminal wealth with the realized volatility of the portfolio, we could capture the intertemporal risk factor more easily with just one parameter, that is $\varepsilon$. We have assumed that the mean-reversion level of the spread is Markov switching and studied the utility maximization problem under full and partial information, corresponding to the cases where the trader may or may not observe the state of the Markov chain directly.

In the full information setting, we have computed the optimal strategy and characterized the value function up to the unique solution to a system of ODEs via the Feynman--Kac formula. In the partial information case we have first derived the filter dynamics and then studied the corresponding optimization problem, where the unobservable state of the Markov chain is replaced by its filtered estimate. We have addressed the problem by pointwise maximization, and we have represented the value function in terms of the solution of a system of PDEs.

In the last part of the paper, we have presented a numerical example in which the Markov chain has two possible states. In the full information setting, we have studied the behavior of the value function with respect to several parameters. An interesting result has been that the value of the optimal portfolio is always strictly larger than the value function computed when the Markov modulated mean-reversion level of the spread is replaced by the average mean-reversion level (with respect to the stationary distribution of the Markov chain). Hence we have concluded that the knowledge of the average data is not sufficient to obtain the optimal portfolio value.

In the partial information case, we have observed that the trader always benefits from using filtered estimates for the state of the Markov chain instead of the average data. Gains are larger in less uncertain situations. That happens when the conditional probability of being in one of the states is close to zero or one. Correspondingly, when conditional probabilities are close to $\frac{1}{2}$, gains are smaller.
\appendix

\section{Proof of Theorem \ref{optimalpartial} }\label{sec:appendix}

\begin{proof}
Here, as in the case of full information we maximize pointwisely. We first write
\begin{align}\nonumber\mathbb{E}^{t,z,s,\p}[\log Z_T]=&\log(z)+r(T-t)-\mathbb{E}^{t,s,\p}\left[ \int_t^T \frac{h_u^2\eta^2(1+\varepsilon)}{2}\,\ud u\right]     \\
\label{logustint2}&+\mathbb{E}^{t,s,\p}\left[ \int_t^T h_u\left(\kappa (\btheta^\top p_u-S_u)-\frac{\eta^2}{2}+\rho\sigma\eta \right)\ud u\right],\end{align}
where, $\mathbb{E}^{t,s,\p}$ denotes the conditional expectation given $S_t=s$ and $p_t=\p$.
The first and second order conditions imply that the optimal strategy is given by
\begin{equation}
h^{\ast}(t,s,\p)=\frac{1}{1+\varepsilon}\left(\frac{\kappa\left(\btheta^{\top}\p-s\right)}{\eta^2}+\frac{\rho \sigma}{\eta}-\frac{1}{2}\right).
\end{equation}
This leads to the following stochastic representation for the optimal value,
\begin{align}\label{stochrepvalue}
\log(z)+r(T-t)+\mathbb{E}^{t,s,\p}\left[ \int_t^T \frac{(\kappa (\btheta^{\top}p_u-S_u)-\frac{\eta^2}{2}+\rho\sigma\eta )^2}{2\eta^2(1+\varepsilon)}\,\ud u\right].
\end{align}
We define the function $u:[0,T]\times \R\times \Delta_K \to \R_+$ by
\begin{align}
u(t,s,\p)= \mathbb{E}^{t,s,\p}\left[ \int_t^T \frac{(\kappa (\btheta^{\top}p_u-S_u)-\frac{\eta^2}{2}+\rho\sigma\eta )^2}{2\eta^2(1+\varepsilon)}\,\ud u\right].
\end{align}
By applying the Feynman--Kac formula and plugging the ansatz $u(t,s,\p)=d(t)s^2+c(t,\p)s+f(t,\p)$ in the resulting equation leads to the system of linear partial differential equations in  \eqref{eq:pde-1}-\eqref{eq:pde-2} and the following linear ordinary differential equation
\begin{align}\label{eq:D}
d_t(t)-2\kappa d(t)+ \frac{\kappa^2}{2\eta^2(1+\varepsilon)}=0, \quad d(T)=0.
\end{align}
Note that the system \eqref{eq:pde-1} and \eqref{eq:pde-2} admits a unique solution, see Chapter 9 of \citet{friedman2008partial}.
\smallskip

\end{proof}
\section*{Acknowledgments}
S\"{u}han Altay gratefully acknowledges financial support from the Austrian Science Fund (FWF) under grant P25216. The work on this paper was completed while Zehra Eksi was visiting the Department of Economics, University of Perugia as a part of the ACRI Young Investigator Training Program (YITP). The support of the Association of Italian Banking Foundations and Savings Banks (ACRI) is greatly acknowledged.

%\section*{Acknowledgments}
%
%The authors thank the participants of $14^{\text{th}}$ Europt Workshop on Advances in Continuous Optimization for discussions. S\"{u}han Altay gratefully acknowledges financial support from the Austrian Science Fund (FWF) under grant P25216. The work on this paper was completed while Zehra Eksi was visiting the Department of Economics, University of Perugia as a part of the ACRI Young Investigator Training Program (YITP). The support of the Association of Italian Banking Foundations and Savings Banks (ACRI) is greatly acknowledged.
%
%\bibliographystyle{plainnat}
%\bibliography{references_altay}

\end{document}